\documentclass{article}

\usepackage{microtype}%
\usepackage{amsmath,amsthm}

\usepackage{tikz}
\usetikzlibrary{arrows,automata}

\newtheorem{theorem}{Theorem}[section]

\newtheorem{prop}[theorem]{Proposition}
\newtheorem{lemma}[theorem]{Lemma}
\newtheorem{corollary}[theorem]{Corollary}
\theoremstyle{definition}

\theoremstyle{remark}

\newcommand{\FO}{\ensuremath{\mathrm{FO}}}
\newcommand{\DA}{\ensuremath{\mathbf{DA}}}
\newcommand{\Ap}{\ensuremath{\mathbf{Ap}}}
\newcommand{\J}{\ensuremath{\mathbf{J}}}
\newcommand{\V}{\ensuremath{\mathbf{V}}}
\newcommand{\W}{\ensuremath{\mathbf{W}}}
\newcommand{\JI}{\ensuremath{\mathbf{J_1}}}
\newlength{\oldparindent}
\newlength{\indentwidth}
\newcommand{\negindent}[1]{%
   \settowidth{\indentwidth}{\textcolor{darkgray}{#1\ }\ }%
   \setlength{\oldparindent}{\parindent}%
   \setlength{\parindent}{-\indentwidth}%
   \textcolor{darkgray}{#1\ }%
   \setlength{\parindent}{\oldparindent}%
}

\newcommand{\longlongrightarrow}{\relbar\joinrel\longrightarrow}

\title{An effective characterization of the alternation hierarchy in two-variable logic}
\author{Andreas Krebs\thanks{email:\texttt{mail@krebs-net.de}}\and Howard Straubing\thanks{email:\texttt{straubin@cs.bc.edu}\newline\hspace*{0.4cm} Research partially supported by National Science Foundation Grant CCF-0915065}}

\begin{document}

\maketitle

\begin{abstract}
We characterize the languages in the individual levels of the quantifier alternation hierarchy of first-order logic with two variables by identities.
This implies decidability of the individual levels. More generally we show that the two-sided semidirect product of a decidable variety with the variety \J\ is decidable.
\end{abstract}

\section{Introduction}

It has been known for some time (Kamp~\cite{Ka68}, Immerman and Kozen~\cite{ImKo89}) that every first-order sentence over the base $<$ defining properties of finite words is equivalent to one containing only three variables. The fragment $\FO^2[<]$ of sentences that use only two variables, has been the object of intensive study;  Tesson and Th\'erien~\cite{TeTh02} give a broad-ranging survey of the many places in which the class of languages definable in this logic arises. Weis and Immerman~\cite{WeIm09} initiated the study of the hierarchy within $\FO^2[<]$ based on alternation of quantifiers.  They showed, using model-theoretic techniques, that the hierarchy is infinite, but finite for each fixed alphabet.

In~\cite{St11}, the second author provided an algebraic characterization of the levels of the hierarchy, showing that they correspond to the levels of weakly iterated two-sided semidirect products of the pseudovariety \J\ of finite ${\cal J}$-trivial monoids. This still left open the problem of {\it decidability} of the hierarchy: effectively determining from a description of a regular language the lowest level of the hierarchy to which the language belongs. This problem was apparently solved in Almeida-Weil~\cite{AlWe98}, from which explicit identities for the iterated product varieties can be extracted.  However,  an error in that paper called the correctness of these results into question.  Here we show that the given identities do indeed characterize these pseudovarieties. In particular, since it is possible to verify effectively whether a given finite monoid satisfies one of these identities, we obtain an effective procedure for exactly determining the alternation depth of a regular language definable in two-variable logic.

We show more generally that the two-sided semidirect product of a pseudovariety with \J\ as the right-hand factor preserve decidability.  That is, if we have an effective procedure for determining if a given finite monoid belongs to a variety \V, then we have such a procedure for $\V\mathbin{**}\J.$

At several junctures, our proof could have been shortened by appealing to known results about the algebra of finite categories and the topological theory of profinite monoids, which are the principal tools of~\cite{AlWe98}.  For example,  Theorem \ref{thm:localglobal} is really just the bonded component theorem of Tilson~\cite{Ti87} coupled with Simon's Theorem~\cite{Si75} on ${\cal J}$-trivial monoids.  Lemma \ref{lem:sublemma} closely mirrors  the  work of Almeida on the structure of the free profinite ${\cal J}$-trivial monoid~\cite{Al94}.  In order to keep our argument accessible and self-contained, we have chosen to steer clear of these quite technical results.  We do discuss finite categories, but only at the most elementary level.  Avoiding profinite techniques forces us to give explicit size bounds, but these are of independent interest in decidability questions.

We give the necessary preliminaries from algebra in Section~\ref{section:preliminaries}.  Section~\ref{section:categories} is devoted to our fundamental theorem, a category-based characterization of two-sided semidirect products with \J\ as the right-hand factor.  We apply this result in Section~\ref{section:identities} to obtain explicit identities for the levels of the hierarchy, thus solving the decidability problem.  We use these identities in Section~\ref{section:collapse} to give a new proof of the result of Weis and Immerman that the hierarchy collapses for each fixed input alphabet.  Section~\ref{section:decidability} proves the general decidability-preserving result for block products with \J.  %

After we circulated an early draft of this paper, we became aware of a number of related results.  Kufleitner and Weil~\cite{KuWe12}, building on earlier work of theirs~\cite{KuWe09}, independently established the decidability of the levels of the alternation hierarchy, using an entirely different algebraic characterization.  A proof that $\V\mathbin{**}\J$ is 
decidable if $\V$ appears in the unpublished Ph.D. thesis of Steinberg~\cite{St98}.

\section{Preliminaries}\label{section:preliminaries}
While the principal application of our results is in finite model theory, this paper contains no formal logic {\it per se} and is entirely algebraic in content.  The reader should consult~\cite{St11} and~\cite{WeIm09} for the definition of $\FO^2[<]$ and the alternation hierarchy within it.  For our purposes here, they are to be viewed simply as the language classes corresponding to certain varieties of finite monoids, as discussed below.

\subsection{Finite monoids and regular languages}

See the book by Pin~\cite{Pi86} for a detailed treatment of the matters discussed in this subsection and the next; here we give a brief informal review.

A {\it monoid} is a set $M$ together with an associative operation for which there is an identity element $1\in M.$  If $A$ is a finite alphabet, then $A^*$ is a monoid with concatenation of words as the multiplication.  $A^*$ is the {\it free monoid} on $A$: this means that every map $\alpha:A\to M,$ where $M$ is a monoid, extends in a unique fashion to a homomorphism from $A^*$ into $M.$

Apart from free monoids, all the monoids we consider in this paper are finite.  If $M$ is a finite monoid, then for every element $m\in M$ there is a unique $e\in\{m^k:k>1\}$ that is {\it idempotent,} {\it i.e.,} $e^2=e.$  We denote this element $m^{\omega}.$  

If $M, N$ are monoids then we say $M$ {\it divides} $N,$ and write $M\prec N,$ if $M$ is a homomorphic image of a submonoid of $N.$

We are interested in monoids because of their connection with automata and regular languages:  A {\it congruence} on $A^*$ is an equivalence relation $\sim$ on $A^*$ such that $u_1\sim u_2,$ $v_1\sim v_2,$ implies $u_1v_1\sim u_2v_2.$  The classes of $\sim$ then form a monoid $M=A^*/\sim,$ and the map $u\mapsto [u]_{\sim}$ sending each word to its congruence class is a homomorphism.  If $L\subseteq A^*,$ then $\equiv_L,$ the {\it syntactic congruence} of $L,$ is the coarsest congruence for which $L$ is a union of congruence classes.  The quotient monoid $A^*/\equiv_L$ is called the {\it syntactic monoid} of $L$ and is denoted $M(L).$

We say that a monoid $M$ {\it recognizes} a language $L\subseteq A^*$ if there is a homomorphism $\alpha:A^*\to M$ and a subset $X$ of $M$ such that $\alpha^{-1}(X)=L.$  The following proposition gives the fundamental properties linking automata to finite monoids.

\begin{prop}\label{prop:syntactic}
A language $L\subseteq A^*$ is a regular if and only if $M(L)$ is finite.
A monoid $M$ recognizes $L$ if and only if $M(L)\prec M.$
\end{prop}

\subsection{Varieties and identities}

A collection \V\ of finite monoids closed under finite direct products and division is called a {\it pseudovariety} of finite monoids.  (The prefix `pseudo' is there because of the restriction to finite products, as the standard use of `variety' in universal algebra does not carry this restriction.)

Given a pseudovariety \V, we consider for each finite alphabet $A$ the set $A^*\mathcal{V}$ of regular languages $L\subseteq A^*$ such that $M(L)\in\V.$  We call $\mathcal{V}$ the {\it variety of languages} corresponding to the pseudovariety \V.  The correspondence $\V\mapsto\mathcal{V}$ is one-to-one, a consequence of the fact that every pseudovariety is generated by the syntactic monoids it contains. We are interested in this correspondence because of its connection with decidability problems for classes of regular languages:  To test whether a given language $L$ belongs to $A^*\mathcal{V},$ we compute its syntactic monoid $M(L)$ and test whether $M(L)\in\V.$  Since the multiplication table of the syntactic monoid can be effectively computed from any automaton representation of $L,$ decidability for the classes $A^*\mathcal{V}$ reduces to determining whether a given finite monoid belongs to \V.

Let $\Xi$ be the countable alphabet $X=\{x_1,x_2,\ldots\}.$  A {\it term} over $\Xi$  is built from the letters by concatenation and application of a unary operation $v\mapsto v^{\omega}.$  For example, $(x_1x_2)^{\omega}x_1$ is a term.  We will interpret these terms in finite monoids in the obvious way, by considering a valuation $\psi:\Xi\to M$ and extending it to terms by giving concatenation and the $\omega$ operator their usual meaning in $M.$  For this reason, we do not distinguish between $(uv)w$ and $u(vw),$ where $u,v$ and $w$ are themselves terms,nor between terms $u^{\omega}$ and $(u^{\omega})^{\omega},$ as these will be equivalent under every valuation.

An {\it identity} is a formal equation $u=v,$ where $u$ and $v$ are terms.  We say that a monoid $M$ {\it satisfies} the identity, and write $M\models(u=v),$ if $u$ and $v$ are equal under every valuation into $M.$   The family of all finite monoids satisfying a given set of identities is a pseudovariety, and we say that the pseudovariety is {\it defined} by the set of identities. We must stress that the identities we consider here are very special instances of a much more general class of {\it pseudoidentities}.  Under this broader definition, every pseudovariety is defined by a set of pseudoidentities. See, for instance, Almeida~\cite{Al94}.  If a pseudovariety \V\ is defined by a {\it finite} set of identities of the form we described, then membership of a given finite monoid $M$ in \V\ is decidable, since we only need substitute elements of \V\ for the variables in the identities in every way possible, and check that equality holds in each case. 

We consider four particular pseudovarieties that will be of importance in this paper.  (In presenting identities we will relax the formal requirement that all terms are over the alphabet $\{x_1,x_2,\ldots\},$ and use a larger assortment of letters for the variables.)

\medskip

\leftskip=1cm
\negindent{\Ap} The pseudovariety \Ap\ consists of the {\it aperiodic} finite monoids, those that contain no nontrivial groups.  It is defined by the identity $x^{\omega}=xx^{\omega}.$  If $A$ is a finite alphabet and $L\subseteq A^*$ is a regular language, then $M(L)\in\Ap$ if and only if $L$ is definable by a first-order sentence over $<.$ In other words, the first-order definable languages form the variety of languages corresponding to {\bf Ap.}

\smallskip

\negindent{\DA}
The pseudovariety \DA\ is defined by the pair of identities
$$(xyz)^{\omega}y(xyz)^{\omega}=(xyz)^{\omega}.$$ 
There are many equivalent characterizations of this pseudovariety in terms of other identities, the ideal structure of the monoids, and logic.  For us the most important ones are these:  First, \DA\ is also defined by the identities
$$(xy)^{\omega}(yx)^{\omega}(xy)^{\omega}=(xy)^\omega, x^{\omega}=xx^{\omega}.$$
Second, let $e\in M$ be idempotent, and let $M_e$ be the submonoid of $M$ generated by the elements $m\in M$ for which $e\in MmM.$  Then $M\in\DA$ if and only if $e=eM_ee$ for all idempotents $e$ of $M.$  Finally,  if $L\subseteq A^*$ is a regular language, then $M(L)\in\DA$ if and only if $L$ is definable in $\FO^2[<].$ In other words, the two-variable definable languages form the variety of languages corresponding to \DA.

\smallskip

\negindent{\J}
The pseudovariety $\J$ consists of finite monoids that satisfy the pair of identities
$$(xy)^{\omega}=(yx)^{\omega}, x^{\omega}=xx^{\omega}.$$

This is equivalent to the identities
$$(xy)^{\omega}x = y(xy)^{\omega}= (xy)^{\omega}.$$ 
Alternatively, \J\ consists of finite monoids $M$ such that for all $s,t\in M,$ $MsM=MtM$ implies $s=t.$  Such monoids are said to be ${\cal J}$-trivial.

A theorem due to I. Simon~\cite{Si75} describes the regular languages whose syntactic monoids are in \J. Let $w\in A^*.$   We say that $v=a_1\cdots a_k,$ where each $a_i\in A,$ is a {\it subword} of $w$ if
$$w=w_0a_1w_1\cdots a_kw_k$$
for some $w_i\in A.$   We define an equivalence relation $\sim_k$ on $A^*$ that identifies two words if and only if they contain the same subwords of length no more than $k.$  (In particular, $w_1\sim_1 w_2$ if and only if $w_1$ and $w_2$ contain the same set of letters.)  Simon's theorem is:
\begin{theorem}\label{theorem.simon}
\leftskip=1cm
Let $\phi:A^*\to M$ a homomorphism onto a finite monoid.  Then the following are equivalent: 
\begin{itemize}
\leftskip=1cm
\item $M \in\J.$
\item There exists $k\geq 1$ such that if $w\sim_k w',$ then $\phi(w)=\phi(w').$ (In particular, $M$ is a quotient of $A^*/\sim_k.$)
\end{itemize}
\end{theorem}

It is easy to show that the second condition implies the first; the deep content of the theorem is the converse implication.  The theorem can also be formulated in first-order logic: The variety of languages corresponding to \J\ consists of languages definable by boolean combinations of $\Sigma_1$ sentences over $<$.

\smallskip

\negindent{\JI} 
The pseudovariety $\JI$ consists of all idempotent and commutative monoids; {\it i.e.,} those finite monoids that satisfy the identities $x^2=x, xy=yx.$  A language $L\subseteq A^*$ is in the variety of languages corresponding to $\JI$ if and only if it is a union of $\sim_1$-classes.
It is well known, and easy to show, that
$$\JI\subseteq \J\subseteq\DA\subseteq \Ap,$$
and all the inclusions are proper.

\leftskip=0cm

\subsection{Two-sided Semidirect Products}\label{section.twosidedproducts}

In this section we describe an operation on  pseudovarieties of finite monoids, the {\it two-sided semidirect product.}  This was given its formal description by Rhodes and Tilson~\cite{RhTi89}, but it has precursors in automata theory in the work of Sch\"utzenberger on sequential bimachines~\cite{Sc61}, Krohn, Mateosian and Rhodes~\cite{KrMaRh67}, and Eilenberg on triple products~\cite{Ei76}.  Traditionally, one begins with a two-sided semidirect product operation on monoids, and then uses this to define the corresponding operation on pseudovarieties.  Here we find it simpler to define the operation on varieties directly.

Let $A$ be a finite alphabet, and $\psi:A^*\to N$ a homomorphism into a finite monoid.  Let $\Sigma=N\times A\times N,$ which we treat as a new finite alphabet.  We define a length-preserving transduction (not a homomorphism) $\tau_{\psi}:A^*\to \Sigma^*$ by
$$\tau_{\psi}(a_1\cdots a_n)=\sigma_1\cdots \sigma_n\text{, where}$$
$$\sigma_i=(\psi(a_1\cdots a_{i-1}),a_i,\psi(a_{i+1}\cdots a_k))\in\Sigma.$$
(If $i=1,$  we interpret the right-hand side as $(1, a_1,\psi(a_2\cdots a_k)),$ and similarly if $i=n.$)

Let \V\ and \W\ be pseudovarieties of finite monoids.  Let $M$ be a finite monoid, and let $\phi:A^*\to M$ be a surjective homomorphism.  We say that $M\in\V\mathbin{**}\W$ if and only if there exist homomorphisms
$$\psi:A^*\to N\in \W,$$
$$h: (N\times A\times N)^*\to K\in\V,$$
such that  $\phi$ factors through $(h\circ \tau_{\psi},\psi)$---in other words, for all $v,w\in A^*,$ if $\psi(v)=\psi(w)$ and $h(\tau_{\psi}(v))=h(\tau_{\psi}(w)),$ then $\phi(v)=\phi(w).$  It is not difficult to check that this is independent of the alphabet $A$ and the homomorphism $\phi,$ and is thus determined entirely by $M,$ and that furthermore $\V\mathbin{**}\W$ forms a pseudovariety of finite monoids.  We will treat this as the definition of $\V\mathbin{**}\W,$ but it is also straightforward to verify that this coincides with the pseudovariety generated by two-sided semidirect products $K\mathbin{**}N,$ where $K\in\V$ and $N\in\W.$

We define a sequence $\{\V_i\}_{i\geq 1}$ of pseudovarieties by setting $\V_1=\J,$ and, for $i\geq 1,$ $\V_{i+1}=\V_i\mathbin{**}\J.$  The main result of~\cite{St11} is that \DA\ is the union of the pseudovarieties $\V_i,$ and that the variety of languages corresponding to $\V_i$ is the $i^{\mathrm{th}}$ level of the alternation hierarchy within $\FO^2[<].$

\subsection{Finite categories}
We give a brief account of the tools from the algebraic theory of finite categories needed to prove our main results.  The  original papers of Tilson~\cite{Ti87} and Rhodes and Tilson~\cite{RhTi89} give a complete and careful exposition of the general theory. 

The categories studied in category theory are typically big categories, in which the object class consists of something like all topological spaces, and the arrows are all continuous functions.
 The work of Tilson~\cite{Ti87} showed the utility of studying very small categories in which the object set,  as well as each set of arrows between two objects, is finite.  

A category $\mathcal C$ consists of a set of {\em objects} $\mathrm{obj}(\mathcal C)$, a set of {\em arrows} $\mathrm{hom}(A,B)$ from $A$ to $B$ for all $A,B\in\mathrm{obj}(\mathcal C)$, and a associative partial binary operations $\circ:\mathrm{hom}(A,B)\times\mathrm{hom}(B,C)\rightarrow\mathrm{hom}(A,C)$ for all $A,B,C\in\mathrm{obj}(\mathcal{C})$ called {\em composition}, such that there is an identity in $\mathrm{hom}(A,B)$ for all $A,B\in\mathrm{obj}(\mathcal C)$.

In this view, a finite monoid is simply a category with a single object, and a finite category is consequently a generalized finite monoid.  

Let $A$ be a finite alphabet, $M$ and $N$ finite monoids with homomorphisms 
$$M\stackrel{\phi}{\longleftarrow} A^* \stackrel{\psi}{\longrightarrow} N,$$
where $\phi$ maps onto $M.$ 
We will define a finite category, which we call the {\it kernel category} $\ker({ \psi\circ\phi^{-1}}).$ The {\it objects} of $\ker({\psi\circ\phi^{-1}})$ are pairs $(n_1,n_2)\in N\times N.$\footnote{The odd notation for the kernel  category is used to maintain consistency with the traditional setting for these finite categories. $\psi\circ\phi^{-1}$ is a {\it relational morphism} from $M$ to $N,$ and Tilson defines these categories for arbitrary relational morphisms, not just those derived from morphisms of the free monoid.}

The {\it arrows} are represented by triples
$$(n_1,n_2)\stackrel {u}{\rightarrow} (n_1',n_2'),$$
where $u\in A^*,$ $n_1'=n_1\cdot \psi(u)$ and $\psi(u)\cdot n_2'=n_2.$  Whenever we have a pair of consecutive arrows
$$(n_1,n_2)\stackrel {u}{\rightarrow} (n_1',n_2'), (n_1',n_2')\stackrel {v}{\rightarrow} (n_1'',n_2''),$$
then we can define the product arrow
$$(n_1,n_2)\stackrel {uv}{\longrightarrow} (n_1'',n_2'').$$
If this were all there were to arrows in the kernel category, we would in general have an infinite set of arrows between two objects.  However, we identify two coterminal arrows 
$$(n_1,n_2)\stackrel {u, u'}{\longlongrightarrow} (n_1',n_2')$$
if for all $v,w\in A^*$ with $\psi(v)=n_1,$ $\psi(w)=n_2',$
$$\phi(vuw)=\phi(vu'w).$$
It is easy to check that this identification is compatible with the product on consecutive arrows, so the true arrows of $\ker({\psi\circ\phi^{-1}})$ are equivalence classes modulo this identification. In particular, the finiteness of $M$ and $N$ implies that there are only finitely many distinct arrows

If $(n_1,n_2)=(n_1',n_2'),$ then any pair of arrows from $(n_1,n_2)$ to itself are consecutive, and thus the set of all such arrows at $(n_1,n_2)$ is a finite monoid, which we denote $M_{n_1,n_2}.$  This is a {\it base monoid.}  Base monoids, then, are just built from words $u$ satisfying $n_1\cdot \psi(u)=n_1,$ and $\psi(u)\cdot n_2=n_2,$ and collapsing modulo the equivalence relation identifying arrows.  

The following Lemma concerning the structure of the base monoids will be quite useful.

\begin{lemma}\label{lemma:basemonoids}
Let $A$ be a finite alphabet: $M, N, N'$ finite monoids, and consider homomorphisms
$$M\stackrel{\phi}{\longleftarrow} A^* \stackrel{\psi}{\longrightarrow} N\stackrel{\psi'}{\longrightarrow}N',$$
where $\phi$ maps onto $M.$  Then every base monoid of $\ker({\psi\circ\phi^{-1}})$ divides some base monoid of $\ker({(\psi'\psi)\circ\phi^{-1}}).$
\end{lemma}
\begin{proof}
Let $n_1,n_2\in N.$  We denote by $M_1$ the base monoid at $(n_1,n_2)$ in $\ker({\psi\circ\phi^{-1}}),$
and by $M_2$ the base monoid at $(n_1',n_2')=(\psi(n_1),\psi(n_2))$ in $\ker({(\psi'\psi)\circ\phi^{-1}}).$  
Set
$$U=\{u\in A^*: n_1\cdot\psi(u)=n_1, n_2=\psi(u)\cdot n_2\},$$
$$U'=\{u\in A^*: n_1'\cdot\psi'\psi(u)=n_1', n_2'=\psi'\psi(u)\cdot n_2'\}.$$
$U$ and $U'$ are submonoids of $A^*,$ and $U\subseteq U'.$ $M_1$ and $M_2$ are the quotients of $U$ and $U'$ by the congruences identifying equivalent arrows in the respective categories.  Let $u,u'\in U$ represent equivalent arrows of $M_2,$ and suppose $v,w\in A^*$ are such that $\psi(v)=n_1,$ $\psi(w)=n_2.$ Then $\psi'\psi(v)=n_1',$ $\psi'\psi(w)=n_2',$ so by equivalence in $M_2$ we have $\phi(vuw)=\phi(vu'w).$  But this means that $u$ and $u'$ represent equivalent arrows in $M_1,$ so $M_1$ is a quotient of the image of $U$ in $M_2.$  Thus $M_1\prec M_2.$
\end{proof}

It is worth keeping in mind the somewhat counterintuitive message of this lemma:  The category $\ker({\psi\circ\phi^{-1}})$ is {\it bigger} (it has more objects) than $\ker({(\psi'\psi)\circ\phi^{-1}})$ but its base monoids are {\it smaller}.

The reason for the construction of the kernel category is its relation to two-sided semidirect products.  Roughly speaking, $M\in\V\mathbin{**}\W$ if and only if there exists $\psi:A^*\to N\in\W$ such that 
the category $\ker (\psi\circ\phi^{-1})$ is `globally in \V'.  We will not define this precisely, but instead prove the consequence that if $M\in\V\mathbin{**}\W,$ then $\ker (\psi\circ\phi^{-1})$ satisfies a weaker condition of being `locally in \V.'

\begin{prop}\label{proposition:local}
Let $\phi:A^*\to M$ be a homomorphism mapping onto $M.$  If $M\in\V\mathbin{**}\W,$ then there is a homomorphism $\psi:A^*\to N\in \W$ such that each base monoid of $\ker (\psi\circ\phi^{-1})$  is in \V.
\end{prop}
\begin{proof}
Since $M\in\V\mathbin{**}\W,$ there exist homomorphisms $\psi:A^*\to N\in \W$ and $h:\Sigma^*\to K\in\V,$ where $\Sigma=N\times A\times N$ satisfying the conditions in the definition of the product variety.  Let $n_1,n_2\in N,$ and let $M'$ denote the base monoid at $(n_1,n_2).$ As in the proof of Lemma~\ref{lemma:basemonoids}, we set
$$U=\{u\in A^*: n_1\cdot\psi(u)=n_1, n_2=\psi(u)\cdot n_2\}.$$
If $\sigma=(n,a,n')\in\Sigma,$ we set 
$$\sigma^{n_2}=(n,a,n'n_2)\in\Sigma,$$
$$^{n_1}\sigma=(n_1n,a,n')\in\Sigma,$$
$$ ^{n_1}\sigma^{n_2}=(n_1n,a,n'n_2)\in\Sigma.$$

 We extend these actions to $\Sigma^*$:  If $z=\sigma_1\cdots \sigma_n\in\Sigma^*,$ then $ z^{n_2}={\sigma_1}^{n_2}\cdots{\sigma_n}^{n_2},$ and similarly for the other two operations.
For $u\in U$ we set $\alpha(u)=h(^{n_1}\tau_{\psi}(u)^{n_2}).$  We then have $\alpha(uu')=\alpha(u)\alpha(u')$ for $u,u'\in U.$  (Observe that this property does {\it not} hold for arbitrary $u,u'\in A^*,$ but depends on the fact that words in $U$ stabilize $n_1$ on the right and $n_2$ on the left.)  Thus $\alpha$ is a homomorphism from $U$ into $K.$  

Suppose now that $\alpha(u)=\alpha(u').$  Let $v,w\in A^*$ with $\psi(v)=n_1,$ $\psi(w)=n_2.$ Then
$$h(\tau_{\psi}(vuw))=h(\tau_{\psi}(v)^{n_2})\cdot \alpha(u)\cdot h(^{n_1}\tau_{\psi}(w)),$$
and similarly
$$h(\tau_{\psi}(vuw))=h(\tau_{\psi}(v)^{n_2})\cdot \alpha(u')\cdot h(^{n_1}\tau_{\psi}(w)).$$
So $h(\tau_{\psi}(vuw))=h(\tau_{\psi}(vu'w)),$ and $\psi(vuw)=n_1n_2=\psi(vu'w),$ so $\phi(vuw)=\phi(vu'w).$ Consequently $u,u'$ represent the same element of the base monoid at $(n_1,n_2).$ Thus the map taking $u\in U$ to the corresponding element of $M'$ factors through the homomorphism $\alpha,$ so $M'\prec K,$ giving $M'\in\V.$ 
\end{proof}

\section{A local-global theorem for categories}\label{section:categories}

In general, the converse of Proposition~\ref{proposition:local} is false.  This section is devoted to establishing an important instance in which it is true, namely when $\W=\J.$ 

\begin{theorem}\label{thm:localglobal} Let $A$ be a finite alphabet, $M$ and $N$ finite monoids with $N\in\J$ and homomorphisms 
$$M\stackrel{\phi}{\longleftarrow} A^* \stackrel{\psi}{\longrightarrow} N.$$  Suppose \V\ is a pseudovariety of finite monoids with $\JI\subseteq \V.$  If every base monoid of $\ker({\psi\circ\phi^{-1}})$ is in \V, then $M\in \V\mathbin{**}\J.$
\end{theorem}
\begin{proof}

It follows from Theorem~\ref{theorem.simon} that for some $k>0,$ $\psi$ factors through the homomorphism $A^*\to A^*/\sim_k$ identifying two words that have the same subwords up to length $k.$ By Lemma~\ref{lemma:basemonoids} we may assume that $\psi$ {\it is} this homomorphism, and that $N=A^*/\sim_k.$  In particular, if $w\in A^*,$ then we can represent $\psi(w)$ as the set of subwords of $w$ of length no more than $k.$ (When we need to emphasize the dependence of $\psi$ on the chosen subword length, we will write it as $\psi_k.$)

The set ${\cal P}(N\times N)$ of subsets of $N\times N$ forms an idempotent and commutative monoid with union as the operation, and hence belongs to {\bf V.}  
Let $\Sigma= N\times A\times N$ and let $h_U$ be the homomorphism
$$h_U:\Sigma^*\to {\cal P}(N\times N)$$
$$\text{by }\sigma=(P,a,S)\mapsto\{(P,S)\}$$
for each $\sigma\in\Sigma.$  Given $P,S\in N,$ define a homomorphism
$$h_{P,S}:\Sigma^*\to M_{P,S}$$
by mapping $(P',a,S')\in\Sigma$ to the arrow class of $(P,a,S)$ if $P=P', S=S',$ and to $1\in M_{P,S}$ otherwise.  Finally, set $M'$ to be the direct product
$$M'= {\cal P}(N\times N)\times \prod_{(P,S)\in N\times N} M_{P,S},$$
and set
$$h=h_U \times \prod_{(P,S)\in N\times N}h_{P,S}.$$
By our hypothesis $M'\in \V.$

Let $w,w'\in A^*,$ with
$\psi(w)=\psi(w')$ and  $h(\tau_{\psi}(w))=h(\tau_{\psi}(w')).$ 
We will show $\phi(w)=\phi(w'),$
 which gives the result.

We will look at the paths through $ker (\psi\circ\phi^{-1})$ traced out by $w$ and $w'.$. Since $\psi(w)=\psi(w'),$  the two paths are coterminal, beginning at the object
$(1,\psi(w))$
and ending at
$(\psi(w),1).$
 If the $i^{\mathrm{th}}$ letter of $w$ is $a_i,$ then the $i^{\mathrm{th}}$ arrow on this path is the class of 
$$(P_{i-1},S_{i-1})\stackrel{a_i}{\longrightarrow}(P_i,S_i),$$
where $P_j$ is $\psi(u)$ for the prefix $a_1\cdots a_{j-1}$ of  $j-1$ of $w,$ and likewise $S_j=\psi(v)$ for the suffix $v=a_{j+1}\cdots a_{|w|}.$  
If
$$(P,S)\stackrel{a_i}{\longrightarrow}(P',S')$$
is on the path traced by $w,$ then we have $P\subseteq P'$ and $S'\subseteq S.$  Either $P=P'$ and $S=S',$ in which case this arrow belongs to one of the base monoids, or at least one of the inclusions is proper.  Since $h_U(\tau(w))=h_U(\tau(w')),$ the same pairs $(P,S), (P',S')$ must occur in the path traced by $w'.$  Because of the inclusions, they must occur in the same relative order in this path, with $(P,S)$ preceding $(P',S').$  They also must be adjacent in this path, since if there were a third pair $(P'',S'')$ between them, we would have
$$P\subseteq P''\subseteq P', S'\subseteq S''\subseteq S,$$
so this new pair would have to occur in the original path traced by $w,$ strictly between $(P,S)$ and $(P',S').$  Finally, the letter $a$ labeling the arrow joining these two objects in the respective paths is completely determined by $(P,S)$ and $(P',S').$ This is because at least one of the two inclusions $P\subseteq P'$ and $S'\subseteq S$ is proper.  Assume without loss of generality that the first of these is a proper inclusion.  Then $P'$ contains a word that is not in $P,$ and the last letter of this word is $a.$

Thus our two paths are depicted by the diagram below:
{\small
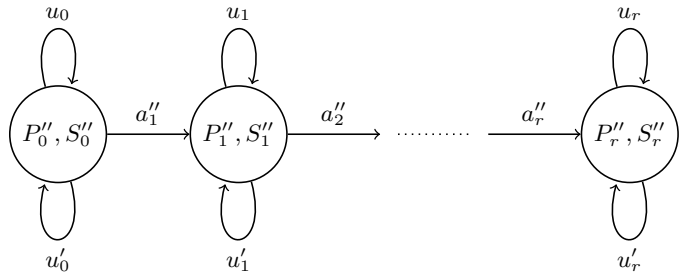
\begin{center}
\begin{tikzpicture}[->,>=to,auto,semithick, node distance=2.4cm]
  \node[state] (A)  {$P_0'',S_0''$};
  \node[state,right of=A] (B)  {$P_1'',S_1''$};
  \node[right of=B,node distance=2.0cm] (C)  {};
  \node[right of=C,node distance=1.2cm] (D)  {};
  \node[state,right of=D,node distance=2.0cm] (E)  {$P_r'',S_r''$};

  \path (A) edge  node {$a_1''$} (B);
  \path (A) edge [loop above]  node {$u_0$} (A);
  \path (A) edge [loop below]  node {$u_0'$} (A);
  \path (B) edge  node {$a_2''$} (C);
  \path (B) edge [loop above]  node {$u_1$} (B);
  \path (B) edge [loop below]  node {$u_1'$} (B);
  \path (C) edge [-,style=dotted] node {} (D);
  \path (D) edge node {$a_r''$} (E);
  \path (E) edge [loop above]  node {$u_r$} (E);
  \path (E) edge [loop below]  node {$u_r'$} (E);

\end{tikzpicture}
\end{center}
}
    
The  paths traverse exactly the same sequence of distinct objects
$$(1,\psi(w))=(P_0'',S_0''), (P_1'',S_1''),\ldots,(P_r'',S_r'')=(\psi(w),1).$$
The arrow joining $(P_{j-1}'',S_{j-1}'')$ and $(P_{j}'',S_{j}'')$ in both these paths is the same letter $a_j''.$  For $j=0,\ldots,r$ each path contains a loop at $(P_j'',S_j'')$ labeled by a factor $u_j$ of $w$ in one path, and a factor $u_j'$ in the other path. We have
$$w = u_0a_1'' u_1\cdots a_r''u_r,$$
$$w' = u_0'a_1'' u_1'\cdots a_r''u_r'.$$
Let $w_0=w,$  and for $j=0,\ldots,r,$ let 
$$w_{j+1}=u_0'a_1''\cdots u_{j}'a_{j+1}''u_{j+1}\cdots a_ru_r,$$
so that $w_{r+1}=w'.$
In other words, we transform $w$ into $w'$ one step at a time, changing each $u_j$ in succession to $u_j'.$  We claim that at each step, $\phi(w_{j})=\phi(w_{j+1}),$ so that we will get $\phi(w)=\phi(w'),$ as required.  Let $(P,S)=(P_j'',S_j'').$  Then, by the definition of the homomorphisms $h_{P,S},$ we have $h_{P,S}(\tau_{\psi}(w))=h_{P,S}(\tau_{\psi}(u_j)),$ $h_{P,S}(\tau_{\psi}(w'))=h_{P,S}(\tau_{\psi}(u_j')),$ and thus
$h_{P,S}(\tau_{\psi}(u_j))=h_{P,S}(\tau_{\psi}(u_j')).$
This means that
$$(P,S)\stackrel{u_j,u_j'}{\longlongrightarrow}(P,S)$$
are equivalent arrows.  Thus
\begin{eqnarray*}
\phi(w_j)=& \phi(u_0'a_1''\cdots u_{j-1}'a_j''u_ja_{j+1}''\cdots u_r)& \\
         =&\phi(u_0'a_1''\cdots u_{j-1}'a_j''u_j'a_{j+1}''\cdots u_r)&=\phi(w_{j+1}).
\end{eqnarray*}
\end{proof}

We remark that the hypothesis $\JI\subseteq \V$ is actually not necessary.  If \V\ is a pseudovariety of monoids that does not contain $\JI,$ then every member of \V\ is a group, and it is know that the converse of Proposition~\ref{proposition:local} holds when \V\ contains only groups; this follows from results in~\cite{Ti87}. We do not require this fact in our main application to the alternation hierarchy.

\section{Effective characterization of levels of the alternation hierarchy}\label{section:identities}

We now define a sequence of identities that will allow us to characterize the varieties $\V_n$.
We set
$$u_1=(x_1x_2)^{\omega}, v_1=(x_2x_1)^{\omega},$$
and for $n\geq 1,$
$$u_{n+1}=(x_1\cdots x_{2n}x_{2n+1})^{\omega}u_n(x_{2n+2}x_1\cdots x_{2n})^{\omega},$$
$$v_{n+1}=(x_1\cdots x_{2n}x_{2n+1})^{\omega}v_n(x_{2n+2}x_1\cdots x_{2n})^{\omega}.$$

\begin{theorem}\label{thm:identities} Let $n\geq 1.$   $M\in\V_n$ if and only if
$M\models (u_n=v_n),$ and $M\models (x^{\omega}=xx^{\omega}).$
\end{theorem}

As we remarked above, when a pseudovariety $\V$is defined by a finite set of identities of this type, one can decide membership in \V. Since the levels of the alternation hierarchy in $\FO^2[<]$ are the varieties of languages corresponding to the $\V_i,$ the alternation depth of a language in $\FO^2[<]$ is effectively computable.

\begin{proof}
The `only if' part (the identities hold in $\V_n$) is proved in~\cite{St11}, so we will just give the proof of the `if' part (sufficiency of the identities).  

We prove the theorem by induction on $n$.
It is well known that the identities $u_1=v_1, x^{\omega}=xx^{\omega}$ characterize $\V_1=\J.$  

So we assume $n>1$ and suppose that $M\models (u_n=v_n)$.  We let $\phi,\psi$ be as in the previous section, so that $\phi$ is any morphism mapping onto $M\in\V_n$, and  $\psi$ depends on the choice of a subword length $K$.  We will show that if $K$ is chosen to be a large enough value, then each base monoid $M_{P,S}$ of the category $\ker (\psi\circ\phi^{-1})$ satisfies the identity $u_{n-1}=v_{n-1}.$  By the inductive hypothesis, this implies that each $M_{P,S}$ belongs to $\V_{n-1},$ and thus by Theorem \ref{thm:localglobal}, $M\in \V_{n-1}\mathbin{**}\J=\V_n.$

We let  $x_1,\ldots x_{2(n-1)}$ be any elements of $M_{P,S}.$   Thus, each $x_i$ is represented by a triple
$$(P,S)\stackrel{w_i}{\longrightarrow} (P,S),$$
where $w_i\in A^*,$ $P\cdot\psi(w_i)=P,$ $\psi(w_i)\cdot S= S.$

We construct words $W_{n-1},W'_{n-1}\in A^*$ by replacing each $x_i$ in $u_{n-1}$ (respectively $v_{n-1}$) by $w_i.$ We will think of $\omega$ in these strings as representing a finite exponent $N$
such that $x^N=x^{N+1}$ for all $x\in M,$ and hence for all $x\in M_{P,S}.$  Thus if $n>2,$
$$W_{n-1}=(w_1w_2\cdots w_{2n-3})^{N}W_{n-2}(w_{2n-2}w_1\cdots w_{2n-4})^{N},$$
$$W_{n-1}'=(w_1w_2\cdots w_{2n-3})^{N}W_{n-2}'(w_{2n-2}w_1\cdots w_{2n-4})^{N}.$$
In the special case $n=2,$ we have $W_1=(w_1w_2)^N,$ $W_1'=(w_2w_1)^N.$

Let $z,y\in A^*$ with $\psi(z)=P,$ $\psi(y)=S.$ 
If $w\in A^*,$ we denote by $\alpha(w)$ the set of letters occurring in $A^*.$  We also denote by $B$ the set $\alpha(W_{n-1})=\alpha(W_{n-1}').$  

\begin{lemma}\label{lem:sublemma} If $K>|M|\cdot (|A|^2+|A|)/2,$ then $z$ has a suffix $z'$ with a factorization
$$z'=z_1z_2\cdots z_{|M|}$$
where 
$$B\subseteq \alpha(z_1)=\alpha(z_2)=\cdots=\alpha(z_{|M|}),$$
and, likewise, $y$ has a prefix $y'$ with a factorization
$$y'=y_1y_2\cdots y_{|M|}$$
where
$$B\subseteq \alpha(y_1)=\alpha(y_2)=\cdots=\alpha(y_{|M|}).$$
\end{lemma}

Assuming the lemma, we will now complete the proof of Theorem \ref{thm:identities}.  
Since $M\models (u_n=v_n),$ we obtain $M\models (xy)^{\omega}(yx)^{\omega}(xy)^{\omega}=(xy)^{\omega}$ by setting $x_1=x,$ $x_2=y,$ and $x_k=1$ for $k>2.$  Thus $M\in\DA.$

We can write $z=z''z',$ where
$z'=z_1\dots z_{|M|}$ has a factorization as in Lemma~\ref{lem:sublemma}. By the standard pumping argument, it follows that there are indices $i\leq j$ such that
$\phi(z_1\ldots z_{i-1})\phi(z_i\cdots z_j)=\phi(z_1\ldots z_{i-1}),$
and thus 
$$\phi(z_1\ldots z_{i-1})\phi(z_i\cdots z_j)^{\omega}=\phi(z_1\ldots z_{i-1}).$$

If we now set 
$$e=\phi(z_i\cdots z_j)^{\omega}$$
$$s_{2n-1}=e\cdot \phi(z_{j+1}\cdots z_{|M|})\text{, and}$$ 
$$s_i=\phi(w_i)\text{ for $i<2n-1,$}$$
 we obtain, from the identity $e\cdot M_e\cdot e=e,$
\begin{eqnarray*}
\phi(z') &=& \phi(z_1\dots z_{i-1}) \cdot \phi(z_{j+1}\cdots z_{|M|})\\
&=& \phi(z_1\dots z_{i-1}) e\cdot e \phi(z_{j+1}\cdots z_{|M|})\\
&=& \phi(z_1\dots z_{i-1}) e \cdot \left(\phi(z_{j+1}\cdots z_{|M|}) (s_1\cdots s_{2n-1})^{\omega -1} s_1\cdots s_{2n-2}\right)\cdot e \phi(z_{j+1}\cdots z_{|M|})\\
&=& \phi(z_1\dots z_{i-1}) e \cdot \phi(z_{j+1}\cdots z_{|M|}) (s_1\cdots s_{2n-1})^{\omega}\\
&=& \phi(z' )\cdot (s_1\cdots s_{2n-1})^{\omega}.
\end{eqnarray*}
The third equality above holds because by Lemma~\ref{lem:sublemma} $z_{i+1}\cdots z_j$ contains all the letters that occur in the $z_k$ and the
$w_k,$ and hence all the values we inserted between occurrences of $e$ belong to $M_e.$

Similarly, using the part of Lemma~\ref{lem:sublemma} concerning the prefix of $y$, we find a value $s_{2n}$ such that
$\phi(y')=(s_{2n}s_1\cdots s_{2n-2})^{\omega}\phi(y').$  Since $M\models (u_n=v_n)$ we obtain
\begin{eqnarray*}
\phi(zW_{n-1}y) &=& \phi(z'')\phi(z')\phi(W_{n-1})\phi(y')\phi(y'')\\
&=&\phi(z'')\phi(z')(s_1\cdots s_{2n-1})^{\omega}\phi(W_{n-1})(s_{2n}s_1\cdots s_{2n-2})^{\omega}\phi(y')\phi(y'')\\
&=&\phi(z'')\phi(z')(s_1\cdots s_{2n-1})^{\omega}\phi(W_{n-1}')(s_{2n}s_1\cdots s_{2n-2})^{\omega}\phi(y')\phi(y'')\\
&=& \phi(z'')\phi(z')\phi(W_{n-1})\phi(y')\phi(y'')\\
&=& \phi(zW_{n-1}'y).
\end{eqnarray*}
But this means that $M_{P,S}\models (u_{n-1}=v_{n-1}),$ as we required.
\end{proof}

We now turn to the proof of Lemma~\ref{lem:sublemma}.  
\begin{proof}[Proof of Lemma \ref{lem:sublemma}]
By symmetry, we only need to treat the part concerning the suffix of $z.$   Recall that $\psi(z)=\psi_K(z)$ is the set of subwords of length no more than $K$ in $z,$ and that $\psi(zb)=\psi(z)$ for all $b\in B.$ 

We will show that if $B\subseteq\alpha(z)$ and $\psi_T(zb)=\psi_T(z)$ for all $b\in B,$ where
$$T=|M|\cdot(k^2+k)/2,$$
and $k=|\alpha(z)|,$
then $z$ contains a suffix with the required properties.  This gives the lemma, because $\psi_K(zb)=\psi_K(z)$ implies $\psi_T(zb)=\psi_T(z)$ for any $\alpha(z)\subseteq A.$

 The proof is by induction on $|\alpha(z)|.$
The base case is when $\alpha(z)=B.$  Let $B=\{b_1,\ldots b_r\}.$  By repeated application of $\psi(zb_i)=\psi(z)$ we find $(b_1\cdots b_r)^{|M|},$ which has length $|M||B|\leq  |M|(|B|^2+|B|)/2,$ is a subword of $z.$ 
 The base case occurs when $B=\alpha(z).$  In this case we can simply take $z'=z$ and factor $z=z_1\cdots z_{|M|},$ where each $z_i$ contains one of the factors $b_1\cdots b_r$ as a subword.

We thus suppose that $\alpha(z)=A'$ contains $B$ as a proper subset. Let $N=|A'|.$ We look at the longest subword
$t_1\cdots t_p$ of $z$ such that $\alpha(t_i)=A'.$  We must have $p\geq 1.$  If $p\geq |M|,$ we can again take $z'=z$ and factor $z$ as $z_1\cdots z_{|M|},$ where each $z_i$ contains $t_i$ as a subword.  If $p<|M|,$ we let $s=|A'|$, then we write
$$t_1\cdots t_p= a_1a_2\cdots a_{ps}\text{, and}$$
$$z=z_0a_1z_1\cdots a_{ps}z_{ps}.$$
We further suppose that this factorization represents the leftmost occurrence of $a_1\cdots a_{ps}$ as a subword of $z,$ in other words that $z_{ps}$ has maximum possible length for this property.  Note that $\alpha(z_{ps})$ is a strict subset of $A',$ for otherwise $z$ would have contained a longer subword $t_1\cdots t_{p+1}$ with $\alpha(t_i)=A'.$ Thus $|\alpha(z_{ps})|\leq N-1.$  Set $T=|M|\cdot((N-1)^2+(N-1))/2.$  We must have $\psi_{T}(z_{ps}b)=\psi_T(z_{ps})$ for all $b\in B.$  If not, there is a subword $u$ of $z_{ps}$ of length less than $T$ such that $ub$ is not a subword of $z_{ps}.$  However $t_1\cdots t_pub$ has length no more than 
$$(|M|-1)\cdot N+T<|M|\cdot (N+((N-1)^2+(N-1))/2)=|M|\cdot(N^2+N)/2,$$
and is accordingly a subword of $z,$ and thus there is a strictly earlier occurrence of $t_1\cdots t_p$ as
a subword of $z,$ a contradiction.  We can thus apply the inductive hypothesis to $z_{ps}$ and conclude that $z_{ps}$ contains a suffix of the required type.
\end{proof}

\section{Collapse of the hierarchy}\label{section:collapse}

In the original model-theoretic study of the alternation hierarchy in $\FO^2[<],$ Weis and Immerman~\cite{WeIm09} showed that while the hierarchy is strict, it collapses for each fixed-size alphabet.  An algebraic proof of strictness was given in~\cite{St11}, using the identities that form the subject of the present paper.  Here we use these techniques to prove the collapse result.

\begin{theorem}\label{thm.collapse}
Let $n>0.$ If $M\in\DA$ is generated by $n$ elements, then $M\in\V_n.$
\end{theorem}

In particular for any fixed alphabet the quantifier alternation hierarchy collapses.
\begin{corollary}
Any language over a $k$-letter alphabet definable by a two-variable sentence is definable by one in which the number of quantifier blocks is $k.$ 
\end{corollary}
\begin{proof}[Proof of Theorem \ref{thm.collapse}]
We prove  by induction on $n$ that if $M$ is generated by $n$ elements and $M\models (u_{N}=v_{N})$ for some $N>n,$ then $M\models (u_n=v_n).$   By Theorem \ref{thm:identities}, this implies the result. Every monoid with one generator is commutative, which gives the result for $n=1.$   We now let $n>1,$ and suppose $M\models (u_{N}=v_{N}).$ Let $X_1,\ldots, X_{2n}\in M.$  Consider the valuation that maps each variable $x_i$ to $X_i$ and let $U_j, V_j\in M$ for $j\leq n$ be the resulting valuations of the terms $u_j, v_j.$    We suppose that $M$ is generated by $n$ elements, so that each $X_i$ itself is a product of these generators.

We consider two cases:  If $X_1,\ldots, X_{2n-2}$ can all be written as products of elements of some strict subset of these  $n$ generators, then all the $X_i$ for $i\leq 2n-2$ belong to an $(n-1)$-generated submonoid $M'$ of $M.$
It follows by the inductive hypothesis that $M'\models (u_{n-1}=v_{n-1}),$ and thus
$U_{n-1}=V_{n-1}.$ 

In the second case, $X_1\cdots X_{2n-2}$ can be written  as a product involving all $n$ generators. 
To ease notation we let
$$A=(X_1\cdots X_{2n-1})^{\omega},$$
$$B=(X_{2n}X_1\cdots X_{2n-2})^{\omega}.$$

Then $U_n=A U_{n-1} B=A A U_{n-1} B B$, since $A$ and $B$ are idempotent.
We know $M\in\DA,$ and every $X_i$ (for $i=1,\dots,2N$) is a product of generators appearing in $X_1,\ldots X_{2n-2},$ and hence in $A$ and $B.$  It follows that all generators are in $M_A$ and $M_B$. We can repeatedly used the identity $e=e\cdot M_e\cdot e$ to insert any product of generators ---in particular, any $X_j$--- between two occurrences of either of the idempotents $A$ or $B.$  

If we set $X_j=1$ for $j=2n+1,\dots,2N$, we get that $$U_N=(X_1\dots X_{2n})^\omega U_n (X_1\dots X_{2n})^\omega.$$
Thus
\begin{eqnarray*}
U_n&=&AA U_{n-1} BB \\
&=& A(X_1\dots X_{2n})^\omega A U_{n-1} B (X_1\dots X_{2n})^\omega B\\
&=& A(X_1\dots X_{2n})^\omega U_{n} (X_1\dots X_{2n})^\omega B\\
&=& A U_{N} B\\
\end{eqnarray*}
Likewise $V_{n}=A V_N B.$  Since $M\models(u_N=v_N),$ we get $U_n=V_n$.
As the valuation on $x_1,\ldots, x_{2n}$ was arbitrary, we have $M\models (u_n=v_n).$
\end{proof}

\section{General decidability results}\label{section:decidability}

 Here we show that for arbitrary pseudovarieties $\V,$ the operation  $\V\mapsto \V\mathbin{**}\J$ preserve decidability.  This of course implies our result (a consequence of Theorem~\ref{thm:identities}) that the varieties $\V_j$ are all decidable, but Theorem~\ref{thm:identities} is a sharper result, since it gives explicit identities.  As we remarked in the introduction, the general decidability result was originally proved by Steinberg~\cite{St98}, but not previously published.  Our approach has the advantages both of being relatively elementary, and yielding explicit bounds on the complexity of membership testing.  

We suppose that $\phi:A^*\to M$ is a surjective homomorphism onto a finite monoid. Let $N>0,$ we denote by $\ker_N\phi$ the category $\ker (\psi_N\circ\phi^{-1}),$ where $\psi_N$ is the natural projection of $A^*$ onto the quotient $A^*/\sim_N.$  We set
$$K=|M|\cdot (|A|^2+|A|)/2$$ as in the statement of Lemma \ref{lem:sublemma}.  With these notations we have:

 \begin{theorem}\label{thm:decidable} Let \V\ be a pseudovariety of monoids.
$M\in \V\mathbin{**}\J$ if and only if every base monoid of $\ker_K\phi$ is in \V.\end{theorem}

We can effectively compute all the objects and arrow classes of $\ker_K\phi$ from $\phi,$ and we can also take $A=M$ and $\phi$ to be the extension of the identity map on $M$ to $A^*.$ The theorem thus immediately implies

\begin{corollary} If $\V$ is a decidable pseudovariety of finite monoids, then so is $\V\mathbin{**}\J.$ \end{corollary}

The remainder of the section is devoted to the proof of Theorem~\ref{thm:decidable}. Strictly speaking, our argument is complete only in the case where $\JI\subseteq \V,$ but see the remark at the end of Section~\ref{section:categories}, which implies that our proof is valid in all cases. 

\begin{proof}[Proof of Theorem \ref{thm:decidable}] If all the base monoids of $\ker_K\phi$ are in \V, then $M\in\V\mathbin{**}\J$ by Theorem \ref{thm:localglobal}. For the converse, we suppose $M\in\V\mathbin{**}\J.$ Then, again by Theorem \ref{thm:localglobal}, there exists $K'$ such that every base monoid of $\ker_{K'}\phi$ is in \V. 

 If $K'\leq K,$ the desired  result follows directly from Lemma~{lemma:basemonoids}.
 So we assume $K'>K.$ We will need the special properties of the morphisms $\psi_N$ given in Lemma \ref{lem:sublemma}. Let $(P,S)$ be an object of $\ker_K\phi.$ Set
 $$B=\{b\in A: P\cdot \psi_K(b)=P,\psi_K(b)\cdot S=S\}.$$
 The base monoid $M_{P,S}$ is generated by the arrows $(P,S)\stackrel{b}{\rightarrow}(P,S),$ and consequently we obtain a homomorphism $\rho$ from $B^*$ onto $M_{P,S}.$
 
 Let $z,y\in A^*$ be such that $\psi_K(z)=P,$ $\psi_K(y)=S.$  By Lemma \ref{lem:sublemma}, there are factorizations
 $$z=z_0z_1\cdots z_{|M|},
 y=y_{|M|}\cdots y_1y_0,$$ 
 where for all $1\leq i,j\leq |M|,$ $B\subseteq\alpha(z_i)=\alpha(z_j)$ and $B\subseteq\alpha(y_i)=\alpha(y_j).$
 
 We also have, for some $1\leq i<j\leq |M|,$
 $$\phi(z_1\cdots z_i)=\phi(z_1\cdots z_iz_{i+1}\cdots z_j),$$
 so we can insert arbitrarily many copies of $z_{i+1}\cdots z_j$ into $z$ (and likewise into $y$) without changing the value of the word under $\phi.$  Let us do this in such a manner that the resulting words 
 $$z'=z_0z_1\cdots z_r, y'=y_r\cdots y_1y_0$$
 contain $r\geq K'$ of the factors $z_k$ and $y_k$ respectively.  Let $P'=\psi_{K'}(z')$ and $S'=\psi_{K'}(y').$  It follows that $P'\cdot\psi_{K'}(b)=P'$ for every $b\in B,$ for if $c_1\cdots c_{K'-1}b$ is a subword of $z'b$ that uses the final letter of $z'b,$ then we can find an occurrence of $c_1\cdots c_{K'-1}$ contained entirely in the factors $z_0\cdots z_{r-1}$ of $z',$ and consequently an occurrence of $c_1\cdots c_{K'-1}b$ in $z'.$  Likewise $S'=\psi_{K'}(b)\cdot S'$ for every $b\in B.$  It follows that $(P',S')\stackrel{b}{\rightarrow}(P',S')$ represents an arrow of $M_{P',S'}.$
 
 We take {\it all} the objects $(P',S')$ that arise in this manner from representatives $z,y$ of $(P,S)$ and form the direct product $N'$ of the resulting $M_{P',S'}.$  For each $b\in B$ we take the element of $N'$ whose value in each component $M_{P',S'}$ is the arrow represented by $(P',S')\stackrel{b}{\rightarrow}(P',S'),$ and we form the submonoid $N$ of $N'$ generated by these elements.  We thus have a homomorphism $\sigma$ from $B^*$ onto $N.$  
 
 Let $w,w'\in B^*$ with $\sigma(w)=\sigma(w').$  We claim $\rho(w)=\rho(w').$  Indeed, let $z,y\in A^*$ with $\psi_K(z)=P$ and $\psi_K(y)=S.$  We obtain $z',y'$ from $z,y$ as above, and set $P'=\psi_{K'}(z),$ $S'=\psi_{K'}(y).$ $(P',S')$ is one of the objects used to build the direct product $N',$ so $\sigma(w)=\sigma(w')$ implies in particular that 
 $(P',S')\stackrel{w,w'}{\longlongrightarrow}(P',S')$ are equivalent arrows.  Thus
 $$\phi(zwy)=\phi(z'wy')
 = \phi(z'w'y')
 = \phi(zwy).
$$
 
 Thus $(P,S)\stackrel{w,w'}{\longlongrightarrow}(P,S)$ are equivalent arrows, so $\rho(w)=\rho(w'),$ as claimed.  Thus $M_{P,S}$ is a homomorphic image of $N,$  thus a divisor of $N'$ and consequently in $\V.$  This shows that all base monoids of $\ker_K\phi$ are in \V, as required.
\end{proof}

\section{Conclusion}

We have shown that the identities given in~\cite{St11} indeed characterize $\V_n$. There is, or course, a one-sided semidirect product, which has been much more thoroughly studied.  Our results, and their proofs, can all be adapted to one-sided products, with little modification.  In this case, the hierarchy collapses at the second level: ${\bf J}*{\bf J}*{\bf J}={\bf J}*{\bf J}.$  (This fact is not new.  It has long been known that the closure of {\bf J} under one-sided products is the pseudovariety {\bf R} of ${\mathcal{R}}$-trivial monoids, and Brzozowski and Fich~\cite{BrFi80} showed ${\bf R}={\bf J}_1*{\bf J}.$)
 
In their Paper Kufleitner and Weil~\cite{KuWe12} give a completely different characterization of the levels of $\FO^2[<]$. It would be nice to see direct connection between these two approaches.

\noindent{\bf Acknowledgements}
We are grateful to Manfred Kufleitner, Benjamin Steinberg, and Pascal Weil for detailed discussions of this work.

\end{document}